\documentclass[submission]{eptcs}
\providecommand{\event}{DCM 2013} 

\usepackage{pdfpages}
\usepackage{amsfonts}
\usepackage{amsmath}
\usepackage{latexsym}
\usepackage{amssymb}
\usepackage[mathcal]{euscript}
\usepackage{tikz}
\usetikzlibrary{matrix}

\newtheorem{definition}{Definition}
\newtheorem{example}{Example}
\newtheorem{proposition}{Proposition}
\newtheorem{theorem}{Theorem}
\newtheorem{lemma}{Lemma}
\newenvironment{proof}{\paragraph{Proof:}}{\hfill$\square$}

\newcommand{\VS}[1]{} 	
\newcommand{\VL}[1]{#1} 	

\newcommand{\cim}[1]{}

\title{Causal Dynamics of Discrete Surfaces}
\author{Pablo Arrighi
\institute{Universit\'e de Grenoble, LIG, 220 rue de la chimie, 38400 Saint-Martin-d'H\`eres, France}
\institute{Universit\'e de Lyon, LIP, 46 all\'ee d'Italie, 69008 Lyon, France}
\thanks{This work was supported by the French National Research Agency, ANR-10-JCJC-0208 CausaQ grant}
\email{parrighi@imag.fr}
\and
Simon Martiel
\institute{Universit\'e Nice-Sophia Antipolis, I3S, 2000 routes des Lucioles, 06900 Sophia Antipolis, France}
\thanks{This work was supported by the John Templeton Foundation, grant ID 15619}
\thanks{This work was supported by the French National Research Agency, ANR-10-JCJC-0208 CausaQ grant}
\email{martiel@i3s.unice.fr}
\and
Zizhu Wang
\institute{Centre for Quantum Information and Communication, Ecole Polytechnique de Bruxelles}
\institute{Universit\'e Libre de Bruxelles, 50 av. F.D. Roosevelt - CP165/59, 1050 Bruxelles, Belgium}
\email{zizhu.wang@ulb.ac.be}
}
\def\titlerunning{Causal Dynamics of Discrete Surfaces}
\def\authorrunning{P. Arrighi, S. Martiel, Z. Wang}

\begin{document}
\maketitle

\begin{abstract}
We formalize the intuitive idea of a labelled discrete surface which evolves in time, subject to two natural constraints: the evolution does not propagate information too fast; and it acts everywhere the same. 
\end{abstract}

\section{Introduction}

Various generalizations of cellular automata, such as stochastics \cite{ArrighiSCA}, asynchronous \cite{manzoni2012asynchronous} or non-uniform cellular automata \cite{FormentiNONUNI}, have already been studied. In \cite{ArrighiCGD,ArrighiIC,ArrighiCayley,ArrighiCayleyNesme} the authors, together with Dowek and Nesme, generalize Cellular Automata theory to arbitrary, time-varying graphs. I.e. they formalize the intuitive idea of a labelled graph which evolves in time, subject to two natural constraints: the evolution does not propagate information too fast; and it acts everywhere the same. Some fundamental facts of Cellular Automata theory carry through, for instance these "causal graph dynamics" admit a characterization as continuous functions. 

The motivation for developing these Causal Graph Dynamics (CGD) was to ``free Cellular Automata off the grid'', so as to be able to model any situation where agents interact with their neighbours synchronously, leading to a global dynamics in which the states of the agents can change, but also their topology, i.e. the notion of who is next to whom. In \cite{ArrighiCGD,ArrighiIC,ArrighiCayley} two examples of such situations are mentioned. The first example is that of a mobile phone network: mobile phones are modelled as vertices of the graph, in which they appear connected if one of them has the other as a contact. The second example is that of particles lying on a smooth surface and interacting with one another, but whose distribution influences the topology the smooth surface (cf. Heat diffusion in a dilating material, or even discretized General Relativity \cite{Sorkin}). CGD seems quite appropriate for modelling the first situation (or at least a stochastic version of it). 

Modelling the second situation, however, is not a short-term perspective. One of the several difficulties we face is that having freed Cellular Automata off the grid, we can no longer interpret our graphs as surfaces, in general. There are, however, a number of formalisms for describing discretized surfaces which are very close to graphs (Abstract simplicial complexes, CW-complexes\ldots \cite{Hatcher}). These work by gluing triangles alongside so as to approximate any smooth surface. Relying upon these formalisms, can we formalize the idea of a labelled discrete surface which evolves in time --- again subject to the constraints that evolution does not propagate information too fast and acts everywhere the same? Can we achieve this by just modelling each triangle as a vertex, and each gluing of two triangles as an edge, and then evolve the graph according to a CGD?

Notice that one could argue that simplicial complexes are not the simplest objects one could use to represent surfaces: planar graphs may seem more natural to some. Our choice is motivated by two reasons. First, the notion of planar graphs can only be used to represent two-dimensional surfaces, and would be limiting when generalizing to higher dimensions (see further work). Second, in a planar graph, the degree of each vertex is not bounded, and thus such graphs would not fit in our model. In order to change this we would have needed to artificially bound this degree by some constant $d$ and lose the generality of planar graphs.

This paper tackles the question of how to give a rigorous definition of ``Causal Dynamics of Simplicial Complexes'', focussing on the 2D case for now. It investigates whether CGD can be readily adapted for this purpose, i.e. whether CGD can be ``tied up again to discrete 2D surfaces''. It will turn out that this can be done at the cost of two additional restrictions, i.e. a CGD must be rotation-commuting and bounded-star preserving in order to be a valid Causal Dynamics of Simplicial Complexes. The first restriction allows us to freely rotate triangles. The second requirement allows us to map geometrical distances into graph distances.  Both restrictions are decidable. This way of modelling simplicial complexes is similar to combinatorial maps defined in~\cite{lienhardt}.

\section{Complexes as graphs}\label{sec:complexesasgraphs}

{\bf Correspondence.} Our aim is to define a Cellular Automata-like model of computation over {\em $2D$ simplicial complexes}. For this purpose, it helps to have a more combinatorial representation of these complexes, as graphs. The straightforward way is to map each triangle to a vertex, and each facet of the triangle to an edge. The problem, then, is that we can no longer tell one facet from another, which leads to ambiguities (see Fig. \ref{fig:correspondance} {\em Top row.}). 

A first solution is to consider {\em $2D$ coloured simplicial complexes} instead. In these complexes, each of the three facets of a triangle has a different colour amongst $\{a,b,c\}$. Now each triangle is again mapped to a vertex, and each facet of the triangle to an edge, but this edge holds the colours of the facets it connects at its ends (see Fig. \ref{fig:correspondance} {\em Bottom row.}). We recover \cite{ArrighiCGD,ArrighiIC,ArrighiCayley} the following definition.

\begin{figure}\begin{center}
\includegraphics[scale=0.5]{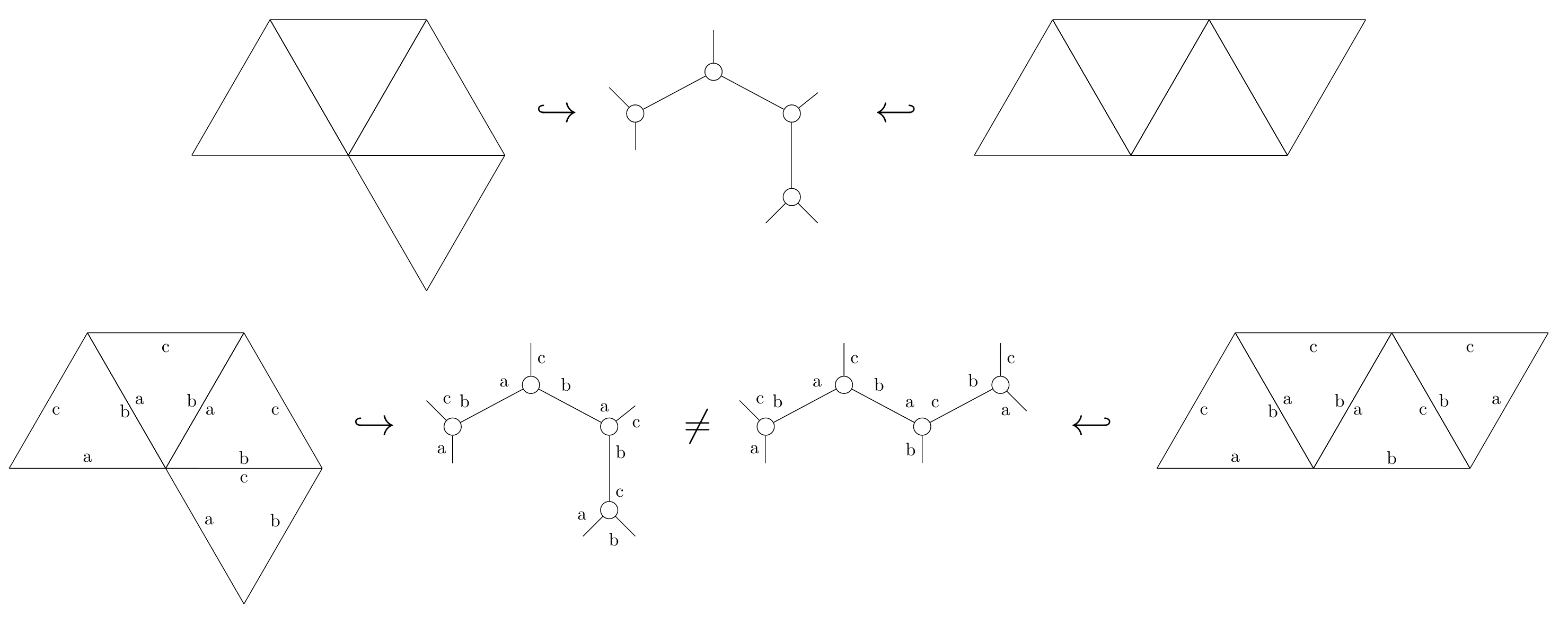}
\caption{Complexes as graphs. \label{fig:correspondance}
{\em Top row}. The straightforward way to encode complexes as graphs is ambiguous.
{\em Bottom row}. Encoding coloured complexes instead lifts the ambiguity. However, the fact that the extreme triangles share one point or not, is less obvious in the graph representation.}
\end{center}\end{figure}

\begin{definition}[Graph]\label{def:graphs}
A labeled {\em graph} $G$ is given by 
\begin{itemize}
\item[$\bullet$] A (at most countable) subset $V(G)$ of $V$, whose elements are called {\em vertices} and where $V$ is the uncountable set of all possible vertex names.
\item[$\bullet$] A finite set $\pi=\{a,b,c\}$, whose elements are called {\em ports}.
\item[$\bullet$] A set $E(G)$ of non-intersecting two element subsets of $V(G):\pi$, whose elements are called edges. Symbol $:$ stands for the cartesian product. An edge $\{u:p,v:q\}$ is to be read ``There is an edge linking port $p$ of vertex $u$ and port $q$ of vertex $v$''.
\item[$\bullet$] A function $\sigma: V(G) \rightarrow \Sigma$ associating to each vertex $v$ some label $\sigma(v)$ in a finite set $\Sigma$.
\end{itemize}
The set of labeled graphs with labels in $\Sigma$ is denoted $\mathcal{G}_{\pi,\Sigma}$ and the set of disks of radius $r$ is denoted $\mathcal{D}^r_{\pi,\Sigma}$.
We similarly define the set of (unlabelled) graphs and denote it $\mathcal{G}_{\pi}$. 
\end{definition}

The following provides a formal interpretation of those graphs into $CW$-complexes.
\begin{definition}[Interpretation]\label{def:interpretation}
Given a graph $G$, its interpretation as a $CW$-complex $K(G)$ is such that: 
\begin{itemize}
\item[$\bullet$] its set of triangles $K_2$ is $V(G)$.
\item[$\bullet$] its set of segments $K_1$ is the quotient of $V(G):\pi$ with respect to the equivalence: $u:p\equiv_1 v:q$ if and only if $\{u:p,u:q\}\in E(G)$. Elements of $K_1$ are denoted $u:\overline{p}$, to distinguish them from the following:
\item[$\bullet$] its set of points $K_0$ is the quotient of $V(G):\pi$ with respect to the equivalence: $u:p\equiv_0 v:q$ if and only if $\{u:(p+1),v:(q-1)\}\in E(G)$.
\end{itemize}
A segment $u:\overline{p}$ has points $\{ u:(p+o) \textrm{ modulo }\equiv_0 \,|\; o\in\{1,2\}\,\}$.\\
A triangle $u$ has segments $\{ u:\overline{p} \textrm{ modulo }\equiv_1\,|\; p\in\pi\,\}$.
Notice that segments $u:\overline{p}$ and $u:\overline{q}$ have common point $u:\overline{p}\cap\overline{q}$.  
\end{definition}

\begin{figure}[h]\begin{center}
\includegraphics[scale=1.0]{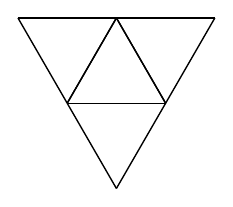}\includegraphics[scale=1.0]{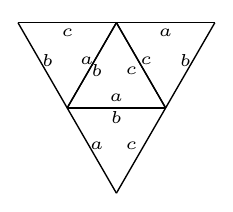}\includegraphics[scale=1.0]{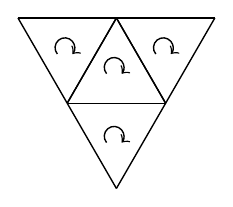}
\caption{Complexes, Coloured complexes, Oriented Complexes \label{fig:complexes}}
\end{center}\end{figure}

This notion of coloured simplicial complex is not so common, however. It is more common to consider a version of coloured complexes where triangles can rotate freely, i.e. where we can permute the colours: $a$ for $b$, $b$ for $c$, $c$ for $a$, so that each triangle has a cyclic ordering of its facets but no privileged facet $a$. The cyclic ordering is then interpreted an orientation: when two facets are glued together in the complex, their orientation must be opposed, so that the two adjacent triangles have the same orientation. This leads to {\em oriented $2D$ simplicial complexes}. Fig. \ref{fig:complexes} summarizes the three kinds of $2D$ simplicial complexes we have mentioned. Definition \ref{def:graphs} captured $2D$ coloured complexes as graphs. How can we capture oriented $2D$ simplicial complexes as graphs? 

\VL{First, we}\VS{We need to} define rotations of the vertices of the graphs in a way that corresponds to rotating the triangles of coloured complexes. Namely, vertex rotations simply permute the ports of the vertex, whilst preserving the rest of the graph:
\begin{definition}[Vertex Rotation]
Let $p_{\operatorname{ports}}$ be some cyclic permutation over $\{a,b,c\}$, and $p_{\operatorname{labels}}$ be some bijection from $\Sigma$ to itself such that $p_{\operatorname{labels}}^3=id$. Let $G$ be a graph and $u\in V(G)$ one of its vertices. Then $r_uG=G'$ is such that $V(G')=V(G)$ and:
\begin{itemize}
\item $\{v:i,w:j\}\in E(G) \wedge v\neq u \wedge w\neq u \Leftrightarrow \{v:i,w:j\}\in   E(G')$.
\item $\{u:i,v:j\}\in E(G) \Leftrightarrow \{u:p_{\operatorname{ports}}(i),v:j\}\in   E(G')$.
\item $\sigma'(u)=p_{\operatorname{labels}}(\sigma(u))$, whereas $\sigma'(v)=\sigma(v)$ for $v\neq u$.
\end{itemize}
\end{definition}
(From now on in order to simplify notations we  will drop all labels $\sigma(.)\in \Sigma$, though all the results of this paper carry through to labelled graphs.)
\VL{A rotation sequence $\overline{r}$ is a finite composition of rotations $r_{u_1}, r_{u_2}, \ldots$. Since rotations commute with each other, a rotation sequence can be seen as a multiset, i.e. a set whose elements can appear several times. Hence, the union of two rotation sequences $\overline{r_1}\sqcup\overline{r_2}$ refers to multiset union. Moreover, since $r_u^3=Id$, we can consider that each rotation appears at most two times in a rotation sequence.}

\VL{Second, we define the equivalence relation induced by the rotations. Using this equivalence relation, we can define graphs in which vertices have cyclic ordering of their edges, but no privileged edge $a$. 
\begin{definition}[Rotation Equivalence]
Two graphs $G$ and $H$ are rotation equivalent if there exists a sequence of rotations $\overline{r}$ such that $\overline{r} G=H$. This equivalence relation is denoted $G\equiv H$.
\end{definition}}

\noindent {\bf Who is next to whom?} On the one hand in the world of $2D$ simplicial complexes, two simplices are adjacent if they share a point. On the other hand in the world of graphs, two vertices are adjacent if they share an edge. These two notions do not coincide, as shown in Fig. \ref{fig:correspondance}. The figure also shows that two triangles share a point if and only if their corresponding vertices are related by a monotonous path:
\begin{definition}[Alternating paths]
Let $\Pi=\{a,b,c\}^2$. We say that $u\in\Pi^*$ is a {\em path} of the graph $G$ if and only if there is a sequence $u$ of pairs of ports $q_ip_i$ such that it is possible to travel in the graph according to this sequence, i.e. there exists $v_0,\ldots , v_{|u|}\in V(G)$ such that for all $i\in\{0\ldots |u|-1\}$, one has $\{v_i: q_i,v_{i+1}: p_i\}\in E(G)$, with $u_i=q_ip_i$.
We say that a path $u=q_0p_0\ldots q_{|u|}p_{|u|}$ alternates at $i= 0\ldots |u|-2$ if either $p_i=q_{i+1}+1$ and $p_{i+1}=q_{i+2}-1$, or $p_i=q_{i+1}-1$ and $p_{i+1}=q_{i+2}+1$. A path is {\em $k$-alternating} if and only if it has exactly $k$ alternations. A path is {\em monotonous} if and only if it does not alternate.
\end{definition}
Thus distance one in complexes is characterized by the existence of a $0$-alternating path. More generally, distance $k+1$ in complexes is characterized by the existence of a $k$-alternating path. 
Recall that our aim is to define a CA-like model of computation over these complexes. In CA models, each cell must have a bounded number of neighbours (or a bounded ``star'' in the vocabulary of complexes). This bounded-density of information hypothesis \cite{Gandy} is the first justification for the following restriction upon the graphs we will consider: 
\begin{definition}[Bounded-star Graphs]\label{def:bsgraphs}
A graph $G$ is {\em bounded-star} of bound $s$ if and only if is monotonous paths are of length less or equal to $s$. 
\end{definition}
Notice that the property is stable under rotation. A further justification for this restriction will be given later. 

\section{Causal Graph Dynamics}

We now provide the essential definitions of CGD, through their constructive presentation, namely as localizable dynamics. We will not detail nor explain nor motivate these definitions in order to avoid repetitions with \cite{ArrighiCGD,ArrighiIC,ArrighiCayley}. Still, notice that in \cite{ArrighiCGD,ArrighiIC,ArrighiCayley} this constructive presentation is shown equivalent to an axiomatic presentation of CGD, which establishes the full generality of this formalism. The bottom line is that these definitions capture all the graph evolutions which are such that information does not propagate too fast and which act everywhere the same. 
\VL{   
\begin{definition}[Isomorphism]
An isomorphism is specified by a bijection $R$ from $V$ to $V$ and acts on a graph $G$ as follow:
\begin{itemize}
\item $V(R(G))=R(V(G))$
\item $\{u:k,v:l\}\in E(G) \Leftrightarrow \{R(u):k,R(v):l\}\in E(R(G))$
\end{itemize}
We similarly define the isomorphism $R^*$ specified by the isomorphism $R$ as the function acting on graphs $G$ such that $V(G)\subseteq \mathcal{P}(V.\{\varepsilon,1,...,b\})$ for any bound $b$ as follow:
\begin{itemize}
\item $R^*(\{u.i,v.j,...\})=\{R(u).i,R(v).j,...\}$
\item $V(R^*(G))=R^*(V(G))$
\item $\{u:k,v:l\}\in E(G) \Leftrightarrow \{R^*(u):k,R^*(v):l\}\in E(R^*(G))$
\end{itemize}
\end{definition}
}
\VL{
\begin{definition} [Consistent]
Consider two graphs $G$ and $H$ in ${\cal G}_{\pi}$, they are {\em consistent} if and only if for all vertices $u,v,w$ and ports $k,l,p$ we have:
$$ \{u:k,v:l\} \in E(G) \wedge  \{u:k,w:p\}\in E(H) \Rightarrow w=v \wedge l=p$$
\end{definition}
}
\begin{definition}[Local Rule]
A function $f:\mathcal{D}^r_\pi \rightarrow \mathcal{G}_\pi$ is called a local rule if there exists some bound $b$ such that:
\begin{itemize}
\item For all disk $D$ and $v'\in D$, $v'\in V(f(D)) \Rightarrow v' \subseteq V(D).\{\varepsilon,1,...,b\}$.
\item For all graph $G$ and all disks $D_1,D_2 \subset G$, $f(D_1)$ and $f(D_2)$ are consistent.
\item For all disk $D$ and all isomorphism $R$, $f(R(D))=R^*(f(D))$\VL{.}\VS{, with $R^*(\{u.i,v.j,...\})=\{R(u).i,R(v).j,...\}$.}
\end{itemize}
\end{definition}
\VL{
\begin{definition}[Union]
The union $G\cup H$ of two consistent graphs $G$ and $H$ is defined as follow:
\begin{itemize}
\item $V(G\cup H)=V(G)\cup V(H)$
\item $E(G\cup H)=E(G)\cup E(H)$
\end{itemize}
\end{definition}
}

\begin{definition}[Localizable Dynamics, a.k.a CGD]\cite{ArrighiCGD,ArrighiIC,ArrighiCayley}
A function $F$ from ${\cal G}_{ \pi}$ to ${\cal G}_{ \pi}$ is a {\em localizable dynamics}, or CGD, if and only if there exists $r$ a radius and $f$ a local rule from ${\cal D}^r_{ \pi}$ to ${\cal G}_{ \pi}$ such that for every graph $G$ in ${\cal G}_{\Sigma,r}$, 
$$F(G)=\bigcup_{v\in G} f(G^r_v).$$
\end{definition}

CGD act on arbitrary graphs. To compute the image graph, they can make use of the information carried out by the ports of the input graph. Thus, they can readily be interpreted as ``Causal Dynamics of Coloured Simplicial Complexes''. But what we are really interested in ``Causal Dynamics of Bounded-star Oriented Simplicial Complexes'', which we will call ``Causal Complexes Dynamics'' for short.

\section{Causal Complexes Dynamics}
This section formalizes Causal Complexes Dynamics (CCD). 

\noindent {\bf Rotation-commutating.} First, we will restrict CGD so that they may use the information carried out by ports, but only as far as it defines an orientation. Formally, this means restricting to dynamics which commute with graphs rotations.

\begin{definition}[Rotation-Commuting function]
A function $F$ from $\mathcal{G}_\pi$ to  $\mathcal{G}_\pi$ is rotation-commuting if and only if for all graph $G$ and all sequence of rotations $\overline{r}$ there exists a sequence of rotations $\overline{r}^*$ such that $ F(\overline{r} G)=\overline{r}^* F(G) $. Such an $\overline{r}^*$ is called a conjugate of $\overline{r}$. The definition extends naturally to functions from $\mathcal{D}_\pi$ to $\mathcal{G}_\pi$.
\end{definition}

\VL{
\begin{lemma}
For all finite set of graphs $G_1,...,G_n$ and for all set of rotation sequences $\overline{r_1},...,\overline{r_n}$, if $G_1,...,G_n$ are consistent with each other, and $\overline{r_1} G_1,...,\overline{r_n} G_n$ are consistent with each other, then 
$$ \bigcup_{i\in\{1,...,n\}} \overline{r_i}G_i =\left(\displaystyle{\bigsqcup_{i\in\{1,...,n\}}} \overline{r_i}\right) \bigcup_{i\in\{1,...,n\}} G_i $$
\end{lemma}
\begin{proof}
Notice that we only need to prove this result for the union of two graphs.
Let us consider two graphs $G_1,G_2$ and two rotation sequences $\overline{r_1},\overline{r_2}$ such that $G_1$ and $G_2$ are consistent and $\overline{r_1}G_1,\overline{r_2}G_2$ are consistent. Let us consider some rotation $r_u$ appearing only once in $\overline{r_1}$. There are two possible cases:
\begin{itemize}
\item $r_u \notin \overline{r_2}$: In this case, $r_u$ acts on $G_1\setminus(G_1\cap G_2)$. 
Indeed, if $u\in V(G_2)$ then $r_u G_1$ and $G_2$ can not be consistent as $u$ has been rotated in the first graph and not in the second.
As $u$ only appears in a part of the graph that is left unchanged by the union, we have that $(r_u (\overline{r_1}\setminus r_u) G_1 )\cup (\overline{r_2} G_2)= r_u[(\overline{r_1}\setminus r_u) G_1 \cup \overline{r_2} G_2] $.
\item $r_u \in \overline{r_2}$: In this case, the two graphs $(\overline{r_1}\setminus r_u)G_1$ and $(\overline{r_2}\setminus r_u)G_2$ are consistent and the rotations sequences $(\overline{r_1}\setminus r_u)$ and $(\overline{r_2}\setminus r_u)$ leave the vertex $u$ unchanged. It is easy to check that the graphs $r_u\left[ (\overline{r_1}\setminus r_u)G_1\cup(\overline{r_2}\setminus r_u)G_2     \right]$ and $\overline{r_1}G_1\cup\overline{r_2}G_2$ are the same.
\end{itemize}
The case where $r_u$ appears more than once in $\overline{r_1}$ can be proven similarly.
By commuting all the rotations with the $\cup$ operator, we have that:
$$ (\overline{r_1}\sqcup \overline{r_2})(G_1\cup G_2)= (\overline{r_1}G_1)\cup(\overline{r_2}G_2) $$
\end{proof}
}

The next question is ``When is a CGD rotation-commuting?''. More precisely, can we decide, given the local rule $f$ of a CGD $F$, whether $F$ is rotation-commuting? The difficulty is that being rotation-commuting is a property of the global function $F$. Indeed, a first guess would be that $F$ is rotation-commuting if and only if $f$ is rotation-commuting, but this turns out to be false. 
\begin{example}
 (Identity function) . Consider the local rule of radius $1$ over graphs of degree $2$ which acts as the identity in every cases but those given in Fig. \ref{fig:noncom}. Because of these two cases, the local rule makes use the information carried out by the ports around the center of the neighbourhood. It is not rotation-commuting. Yet, the CGD it induces is just the identity, which is trivially rotation-commuting.
\begin{figure}
\begin{center}
\includegraphics[scale=0.8]{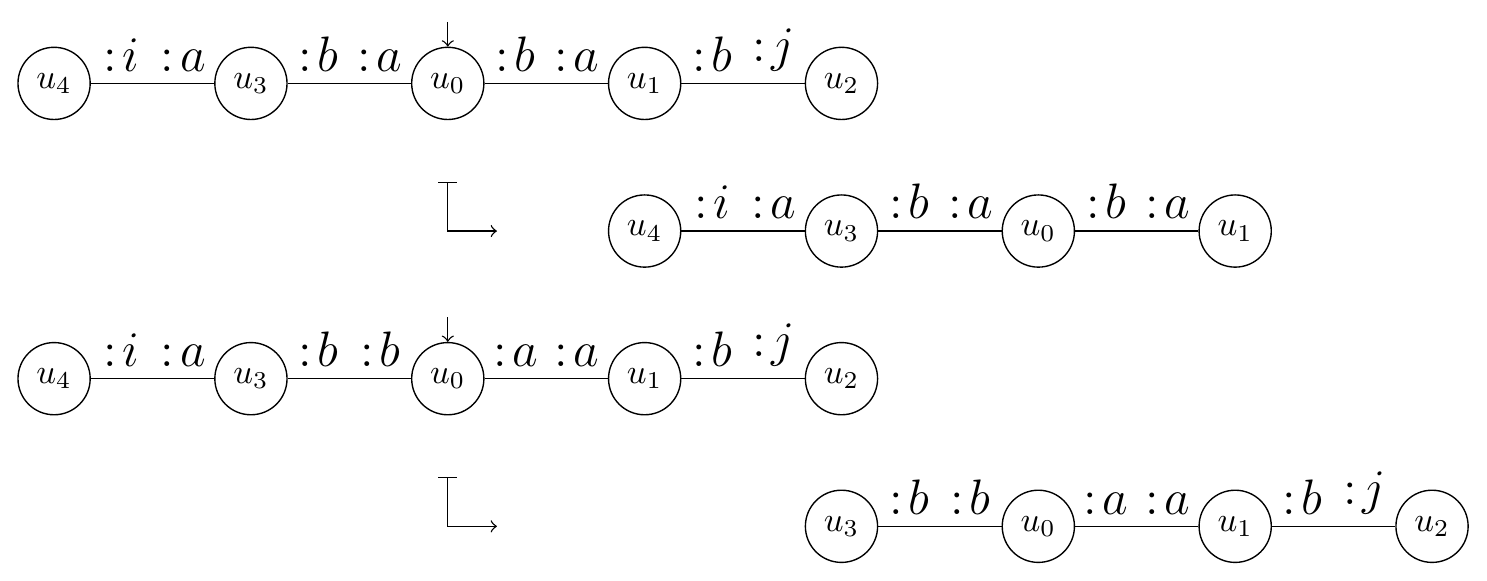}
\end{center}
\caption{A non-rotation commuting local rule induces a rotation commuting CGD.}\label{fig:noncom}
\end{figure}
\end{example}
Thus, unfortunately, some rotation-commuting $F$ can be induced by a non-rotation-commuting $f$.
Yet, fortunately, any rotation-commuting $F$ can be induced by a rotation-commuting $f$.
\begin{theorem}
Let $F$ be a localizable dynamics. $F$ is rotation-commuting if and only if there exists a rotation-commuting local rule $f$ which induces $F$.
\end{theorem}
\VL{
\begin{proof}
$[\Leftarrow]$ Let us consider a rotation-commuting local rule $f$ of radius $r$ inducing a localizable dynamics $F$. Let $G$ be a graph and $u$ a vertex of $G$. The following sequence of equalities proves that $F$ is rotation-commuting:

\[
\begin{array}{lclr}
   F(\overline{r} G) & = &\displaystyle{\bigcup_{v\in G}} f(\overline{r} G^r_v) &\\
   & = & \displaystyle{\bigcup_{v\in G}} \overline{r}_v^* f(G^r_v) &\textrm{  (using $f$ rotation-commuting) }
\end{array}
\]
Using lemma $1$, we can commute the union operator and the sequences of rotations $\overline{r}_v$ as follow:
$$ F(\overline{r} G)= \left(\displaystyle{\bigsqcup_{v\in G}} \overline{r_v}^*\right) \bigcup_{v\in G} f(G^r_v)=\left(\displaystyle{\bigsqcup_{v\in G}} \overline{r_v}^*\right) F(G) $$
Less formally, we can commute rotations and unions by looking at the highest power with which the rotations appear at the right of the union operator.\medskip

\noindent $[\Rightarrow]$ Let $F$ be a rotation-commuting localizable function, and $f$ a local rule inducing $F$. Informally, since $f(G_u^r)$ is included in $F(G)$ we know that as far as orientation is concerned $f$ will indeed be rotation-commuting. However it may still happen that $f$, depending upon the orientation of $G_u^r$, will produce a smaller, or a larger, subgraph of $F(G)$. Therefore, we must define some $\tilde{f}$ which does not do that. Let us consider the following function $\tilde{f}$ from $\mathcal{D}_\pi$ to $\mathcal{G}_\pi$:

$$\forall G_u^r, \tilde{f}(G_u^r)= \bigcup_{ \overline{r} } {\overline{r}^*}^{-1} f(\overline{r} G_u^r)  $$
with $\overline{r}^*$ a conjugate of $\overline{r}$ (given by $F$ rotation-commuting).
\begin{itemize}
\item $\tilde{f}$ is well defined: 
By definition of $\overline{r}^*$ we have that:
\[
\begin{array}{llllr}
 &\forall  \overline{r} ,& f(\overline{r} G_u^r) \subset \overline{r}^* F(G) &\\
 \Rightarrow & \forall  \overline{r}, & {\overline{r}^*}^{-1} f( \overline{r} G_u^r)\subset F(G)& (*)\\
 \Rightarrow & \forall \overline{r_1}, \overline{r_2} , & {\overline{r_1}^*}^{-1} f(\overline{r_1}G_u^r)\ \textrm{and}\ {\overline{r_2}^*}^{-1} f(\overline{r_2}G_u^r)\ \textrm{consistent}&
 \end{array}
\]
\item $\tilde{f}$ is a local rule: we can check that it inherits of the local rule properties of $f$.
\item $\tilde{f}$ induces $F$:
\[\begin{array}{llll}
 \displaystyle{\bigcup_{v\in G}} \tilde{f}(G_u^r)& = & \bigcup_{v\in G} \left[ \displaystyle{\bigcup_{\overline{r}}} {\overline{r}^*}^{-1} f(\overline{r}G_u^r) \right]&\\
 &=& \displaystyle{\bigcup_{v\in G}} \left[ f(G_u^r) \cup \left(\displaystyle{\bigcup_{\overline{r} \neq id}} {\overline{r}^*}^{-1} f(\overline{r}G_u^r)\right)\right]&\\
 &=& F(G)  \cup \displaystyle{\bigcup_{v\in G}}\left(\displaystyle{\bigcup_{\overline{r} \neq id}} {\overline{r}^*}^{-1} f(\overline{r}G_u^r)\right)& \\
 &=F(G)& &\textrm{since $(*)$}
\end{array}\]
\item $\tilde{f}$ is rotation-commuting: let us consider a sequence of rotations $\overline{s}$. We have:
$$\tilde{f}(\overline{s} G_u^r )= \bigcup_{\overline{r} } {\overline{r}^*}^{-1} f(\overline{r}\overline{s}G_u^r)$$
Let us define $\overline{t}=\overline{r}\overline{s}$. As $\overline{r}$ spans all rotations sequences, $\overline{t}$ spans all rotations sequences. We can write:

\[
\begin{array}{lllr}
 \tilde{f}(\overline{s}G_u^r)& =& \displaystyle{\bigcup_{\overline{t}}} {\overline{r}^*}^{-1}f(\overline{t}G_u^r) &\\
	 &=& \displaystyle{\bigcup_{\overline{t}}} {\overline{r}^*}^{-1} \overline{t}^* {\overline{t}^*}^{-1}   f(\overline{t}G_u^r)&\\
	 &=&\left(\displaystyle{\bigsqcup_{\overline{t}}} {\overline{r}^*}^{-1} \overline{t}^*  \right)\displaystyle{\bigcup_{\overline{t}}} {\overline{t}^*}^{-1}   f(\overline{t}G_u^r)&\textrm{using lemma 1}\\
 &=& \left(\displaystyle{\bigsqcup_{\overline{t}}} {\overline{r}^*}^{-1} \overline{t}^*  \right)\tilde{f}(G_u^r)&
\end{array}
\]
\end{itemize}
\end{proof}
}

\begin{proposition}[Decidability of rotation commutation]
Given a local rule $f$, it is decidable whether $f$ is rotation-commuting.
\end{proposition}
\VL{
\begin{proof}
There exists a simple algorithm to verify that $f$ is rotation-commuting.
Let $r$ be the radius of $f$.
We can check that for all disk $D\in\mathcal{D}^r_\pi$ and for all vertex rotation $r_u$, $u\in V(D)$, we have the existence of a sequence $\overline{r}$ such that $f(r_u D)=\overline{r}f(D)$.\\ As the graph $f(D)$ is finite, there is finite number of sequences $\overline{r}$ to test. Indeed, if $|V(f(D))|=k$, we only have $3^k$ different sequences we can apply on $f(D)$ (for each vertex $u$, we can apply $r_u$ 0,1 or 2 times).
Notice that as $f$ is a local rule, changing the names of the vertices in $D$ will not change the structure of $f(D)$ and thus we only have to test the commutation property on a finite set of disks.
\end{proof}
}

\noindent  {\bf Bounded-star preserving.} Second, we will restrict CGD so that they preserve the property of a graph being bounded-star. Indeed, we have explained in Section \ref{sec:complexesasgraphs} that the graph distance between two vertices does not correspond to the geometrical distance between the two triangles that they represent. By modelling CCD via CGD, we are guaranteeing that information does not propagate too fast with respect to the graph distance, but not with respect to the geometrical distance. The fact that the geometrical distance is less or equal to the graph distance is falsely reassuring: the discrepancy can still lead to an unwanted phenomenon as depicted in Fig. \ref{fig:boundedstar}.\\
\begin{figure}\begin{center}
\includegraphics[scale=0.5]{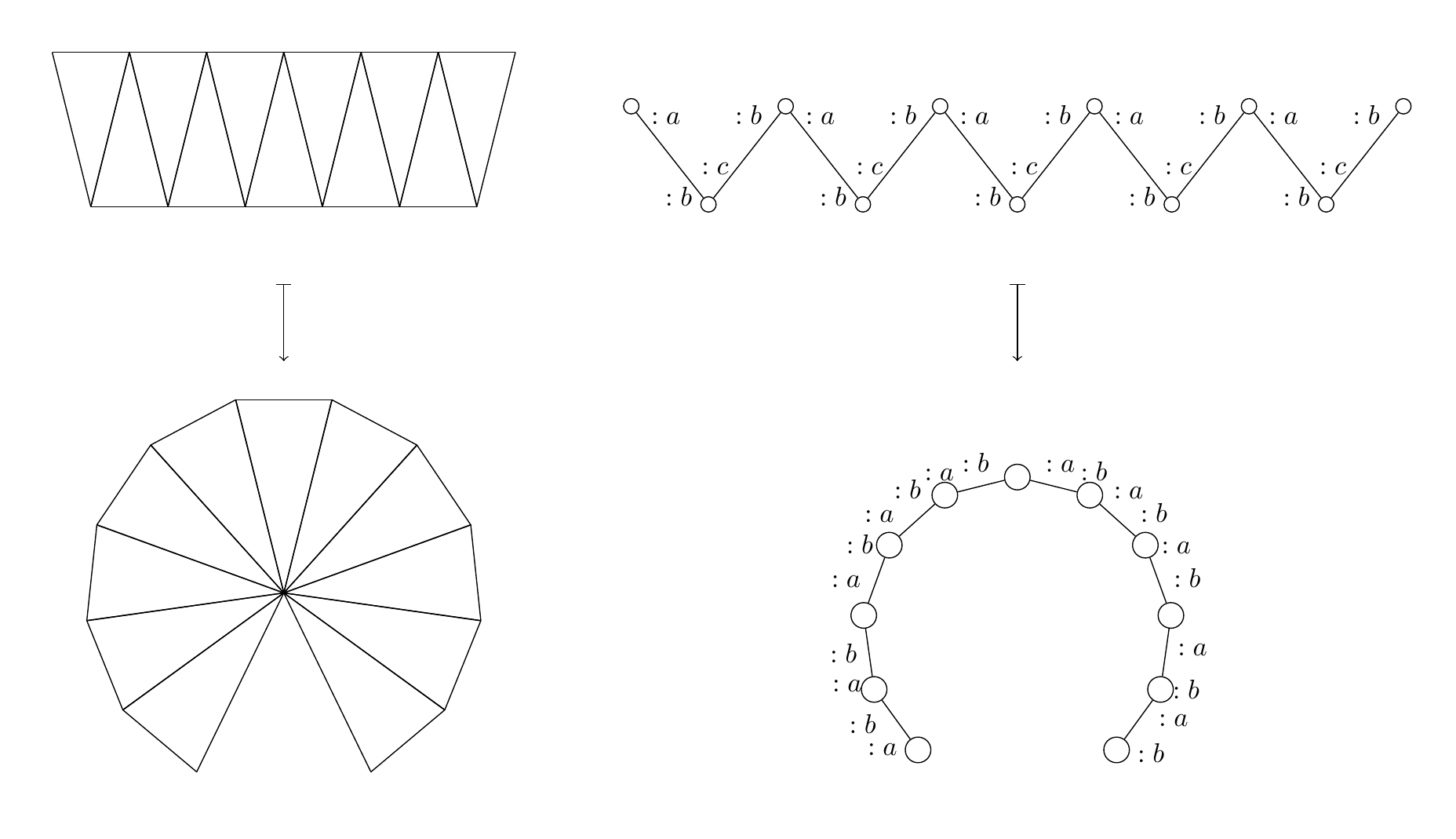}
\caption{An unwanted evolution: sudden collapse in geometrical distance. \label{fig:boundedstar}
{\em Left column:} in terms of complexes. {\em Right column.} In terms of graph representation.}
\end{center}\end{figure}
Of course we may choose not to care about geometrical distance. But if we do care, then we must make the assumption that graphs are bounded-star. This assumption will not only serve to enforce the bounded-density of information hypothesis. It will also relate the geometrical distance and the graph distance by a factor $s$. As a consequence, the guarantee that information does not propagate too fast with respect to the geometrical distance will be inherited from its counterpart in graph distance. In particular, it will forbid the sudden collapse phenomenon of Fig. \ref{fig:boundedstar}. All we need to do, then, is to impose that CCD take bounded-star graphs into bounded-star graphs. This can be decided from its local rule.
\begin{definition}[Bounded-star preserving]
A CGD $F$ is \emph{bounded-star preserving} if and only if for all bounded-star graph $G$, $F(G)$ is also bounded-star.
A local rule $f$ is bounded-star preserving if and only if it induces bounded-star preserving a global dynamics $F$.
\end{definition}
\begin{proposition}[Decidability of bounded-star preservation]
Given a local rule $f$ and a bound $s$, it is decidable whether $f$ is bounded-star preserving with bound $s$.
\end{proposition}
\VL{
\begin{proof}
We can, for each disk $D\in\mathcal{D}^{r}$ centered on a vertex $u$, consider a disk $H$ of radius $2rs$ centered on $u$ and containing $D$. Considering any $0$-alternating path $p$ of $f(D)$, the two following cases can appear:
\begin{itemize}
\item $p$ is strictly contained in $f(D)$ and it can be checked whether its length is greater than $s$,
\item $p$ can be extended as a $0$-alternating path in $f(H)$ by a length between $1$ and $s$. In that case, we can also check if its length is strictly less than $s+1$.
\end{itemize}
By checking this property for each disk $D$ of radius $r$ and each $H$ containing $D$, we can decide whether the image of a graph will contain $0$-alternating path of length greater than $s$.
\end{proof}}

\section*{Conclusion}

\noindent {\bf Summary.} We have obtained that the following definition captures Causal Complexes Dynamics, i.e. evolutions of discrete surfaces such that information does not propagate too fast and that act everywhere the same:
\begin{definition}[Causal Complexes Dynamics]
A function $F$ from ${\cal G}_{ \pi}$ to ${\cal G}_{ \pi}$ is a {\em Causal Complexes Dynamics}, or CCD, if and only if there exists $r$ a radius and $f$ a rotation-commuting, bounded-star preserving local rule from ${\cal D}^r_{ \pi}$ to ${\cal G}_{ \pi}$ such that for every graph $G$ in ${\cal G}_{\Sigma,r}$, $F(G)=\bigcup_{v\in G} f(G^r_v)$.
\end{definition}
We have also obtained that given a candidate local rule $f$, it is decidable whether it has the required properties. Since CCD are a specialization of CGD, several results follow as corollaries from \cite{ArrighiCGD,ArrighiIC,ArrighiCayley}. For instance, it follows that CCD of radius $1$ are universal, that CCD are composable, that CCD can be characterized as the set of continuous functions from discrete surfaces to discrete surfaces with respect to the 
Gromov-Hausdorff-Cantor metric upon isomorphism classes. These results deserve to be made more explicit, but they are already indicators of the generality of the model.

\noindent {\bf Further work.} We went constantly back and forth from graph to simplicial complexes, but we have not formalized this relationship. First: Can every such graph be mapped into a $2D$ oriented simplicial complex? On the one hand, it is intuitive that each vertex represents an oriented triangle, and each edge specifies a unique oriented gluing. On the other hand, we are able to represent a sphere, a cylinder, or a torus with just two vertices, whereas these need many triangles in the simplicial complex formalism. Hence the correspondence is to be found with more economical formalisms such as $\Delta$-complexes \cite{Hatcher}.
\begin{figure}\begin{center}
\includegraphics[scale=0.3]{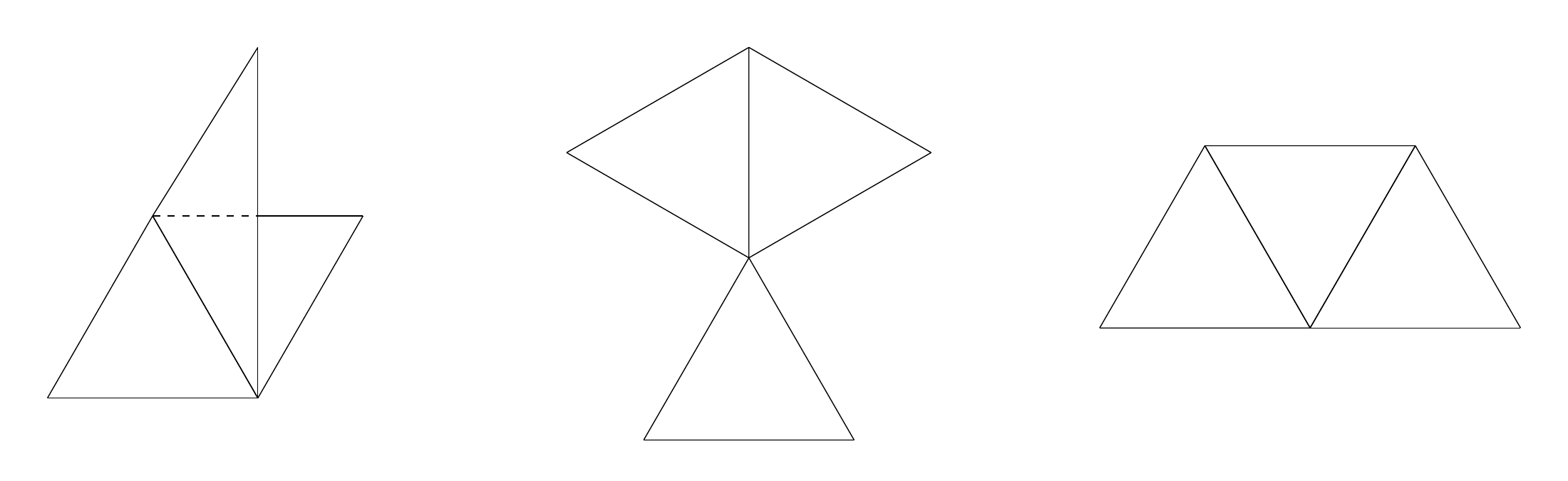}
\caption{Complexes, pseudomanifolds, combinatorial manifolds. \label{fig:pseudo}}
\end{center}
\end{figure}
Second: Can any $2D$ oriented simplicial complex be represented by such a graph? We are willingly limiting ourselves to those complexes that arise as discretizations of $2D$ manifolds, i.e. combinatorial manifolds with borders \cite{Lickorish}. In the $2D$ case these are just the complexes obtained by only gluing triangles pairwise and along their sides (See Fig. \ref{fig:pseudo}). In $n$-dimensions combinatorial manifolds are harder to characterize, however: the star of every point must be an $n$-ball. Our bounded-star restriction will then play a crucial role. 

\section*{Acknowledgements}
The authors would like to thank Christian Mercat for helping them enter the world of simplicial complexes. 

\bibliographystyle{eptcs}
\bibliography{biblio_doi}

\begin{thebibliography}{10}
\providecommand{\bibitemdeclare}[2]{}
\providecommand{\surnamestart}{}
\providecommand{\surnameend}{}
\providecommand{\urlprefix}{Available at }
\providecommand{\url}[1]{\texttt{#1}}
\providecommand{\href}[2]{\texttt{#2}}
\providecommand{\urlalt}[2]{\href{#1}{#2}}
\providecommand{\doi}[1]{doi:\urlalt{http://dx.doi.org/#1}{#1}}
\providecommand{\bibinfo}[2]{#2}

\bibitemdeclare{incollection}{ArrighiCGD}
\bibitem{ArrighiCGD}
\bibinfo{author}{Pablo \surnamestart Arrighi\surnameend} \&
  \bibinfo{author}{Gilles \surnamestart Dowek\surnameend}
  (\bibinfo{year}{2012}): \emph{\bibinfo{title}{Causal Graph Dynamics}}.
\newblock In \bibinfo{editor}{Artur \surnamestart Czumaj\surnameend},
  \bibinfo{editor}{Kurt \surnamestart Mehlhorn\surnameend},
  \bibinfo{editor}{Andrew \surnamestart Pitts\surnameend} \&
  \bibinfo{editor}{Roger \surnamestart Wattenhofer\surnameend}, editors: {\sl
  \bibinfo{booktitle}{Automata, Languages, and Programming}}, {\sl
  \bibinfo{series}{Lecture Notes in Computer Science}} \bibinfo{volume}{7392},
  \bibinfo{publisher}{Springer Berlin Heidelberg}, pp. \bibinfo{pages}{54--66},
  \doi{10.1007/978-3-642-31585-5\_9}.

\bibitemdeclare{article}{ArrighiIC}
\bibitem{ArrighiIC}
\bibinfo{author}{Pablo \surnamestart Arrighi\surnameend} \&
  \bibinfo{author}{Gilles \surnamestart Dowek\surnameend}
  (\bibinfo{year}{2013}): \emph{\bibinfo{title}{Causal Graph Dynamics}}.
\newblock {\sl \bibinfo{journal}{Information and Computation}}
  \bibinfo{volume}{223}, pp. \bibinfo{pages}{78 -- 93},
  \doi{10.1016/j.ic.2012.10.019}.

\bibitemdeclare{inproceedings}{ArrighiCayley}
\bibitem{ArrighiCayley}
\bibinfo{author}{Pablo \surnamestart Arrighi\surnameend} \&
  \bibinfo{author}{Simon \surnamestart Martiel\surnameend}
  (\bibinfo{year}{2012}): \emph{\bibinfo{title}{{Generalized Cayley graphs and
  cellular automata over them}}}.
\newblock In: {\sl \bibinfo{booktitle}{Proceedings of GCM 2012, Bremen,
  September 2012.}}, pp. \bibinfo{pages}{129--143}.
\newblock \urlprefix\url{http://gcm2012.imag.fr/proceedingsGCM2012.pdf}.

\bibitemdeclare{article}{ArrighiCayleyNesme}
\bibitem{ArrighiCayleyNesme}
\bibinfo{author}{Pablo \surnamestart Arrighi\surnameend},
  \bibinfo{author}{Simon \surnamestart Martiel\surnameend} \&
  \bibinfo{author}{Vincent \surnamestart Nesme\surnameend}
  (\bibinfo{year}{2013}): \emph{\bibinfo{title}{{Generalized Cayley graphs and
  cellular automata over them}}}.
\newblock {\sl \bibinfo{journal}{submitted (long version).}}
\newblock \urlprefix\url{http://arxiv.org/abs/1212.0027}.

\bibitemdeclare{inproceedings}{ArrighiSCA}
\bibitem{ArrighiSCA}
\bibinfo{author}{Pablo \surnamestart Arrighi\surnameend},
  \bibinfo{author}{Nicolas \surnamestart Schabanel\surnameend} \&
  \bibinfo{author}{Guillaume \surnamestart Theyssier\surnameend}
  (\bibinfo{year}{2012}): \emph{\bibinfo{title}{Intrinsic Simulations between
  Stochastic Cellular Automata}}.
\newblock In \bibinfo{editor}{Enrico \surnamestart Formenti\surnameend},
  editor: {\sl \bibinfo{booktitle}{{\rm Proceedings 18th international workshop
  on} Cellular Automata and Discrete Complex Systems {\rm and 3rd international
  symposium} Journ\'ees Automates Cellulaires, {\rm La Marana, Corsica,
  September 19-21, 2012}}}, {\sl \bibinfo{series}{Electronic Proceedings in
  Theoretical Computer Science}}~\bibinfo{volume}{90}, \bibinfo{publisher}{Open
  Publishing Association}, pp. \bibinfo{pages}{208--224},
  \doi{10.4204/EPTCS.90.17}.

\bibitemdeclare{article}{FormentiNONUNI}
\bibitem{FormentiNONUNI}
\bibinfo{author}{Alberto \surnamestart Dennunzio\surnameend},
  \bibinfo{author}{Enrico \surnamestart Formenti\surnameend} \&
  \bibinfo{author}{Julien \surnamestart Provillard\surnameend}
  (\bibinfo{year}{2012}): \emph{\bibinfo{title}{Non-uniform cellular automata:
  Classes, dynamics, and decidability}}.
\newblock {\sl \bibinfo{journal}{Information and Computation}}
  \bibinfo{volume}{215}(\bibinfo{number}{0}), pp. \bibinfo{pages}{32 -- 46},
  \doi{10.1016/j.ic.2012.02.008}.

\bibitemdeclare{incollection}{Gandy}
\bibitem{Gandy}
\bibinfo{author}{Robin \surnamestart Gandy\surnameend} (\bibinfo{year}{1980}):
  \emph{\bibinfo{title}{Church's Thesis and Principles for Mechanisms}}.
\newblock In \bibinfo{editor}{H.~Jerome~Keisler \surnamestart
  Jon~Barwise\surnameend} \& \bibinfo{editor}{Kenneth \surnamestart
  Kunen\surnameend}, editors: {\sl \bibinfo{booktitle}{The Kleene Symposium}},
  {\sl \bibinfo{series}{Studies in Logic and the Foundations of Mathematics}}
  \bibinfo{volume}{101}, \bibinfo{publisher}{Elsevier}, pp. \bibinfo{pages}{123
  -- 148}, \doi{10.1016/S0049-237X(08)71257-6}.

\bibitemdeclare{book}{Hatcher}
\bibitem{Hatcher}
\bibinfo{author}{Allen \surnamestart Hatcher\surnameend}
  (\bibinfo{year}{2002}): \emph{\bibinfo{title}{Algebraic Topology}}.
\newblock \bibinfo{publisher}{Cambridge University Press}.
\newblock \urlprefix\url{http://www.math.cornell.edu/~hatcher/AT/ATpage.html}.

\bibitemdeclare{article}{Lickorish}
\bibitem{Lickorish}
\bibinfo{author}{William Bernard~Raymond \surnamestart Lickorish\surnameend}
  (\bibinfo{year}{1999}): \emph{\bibinfo{title}{Simplicial moves on complexes
  and manifolds}}.
\newblock {\sl \bibinfo{journal}{Geometry and Topology Monographs}}
  \bibinfo{volume}{2}, pp. \bibinfo{pages}{299--320},
  \doi{10.2140/gtm.1999.2.299}.
\newblock \urlprefix\url{http://arxiv.org/abs/math.GT/9911256}.

\bibitemdeclare{article}{lienhardt}
\bibitem{lienhardt}
\bibinfo{author}{Pascal \surnamestart Lienhardt\surnameend}
  (\bibinfo{year}{1994}): \emph{\bibinfo{title}{N-dimensional generalized
  combinatorial maps and cellular quasi-manifolds}}.
\newblock {\sl \bibinfo{journal}{International Journal of Computational
  Geometry \& Applications}} \bibinfo{volume}{04}(\bibinfo{number}{03}), pp.
  \bibinfo{pages}{275--324}, \doi{10.1142/S0218195994000173}.

\bibitemdeclare{article}{manzoni2012asynchronous}
\bibitem{manzoni2012asynchronous}
\bibinfo{author}{Luca \surnamestart Manzoni\surnameend} (\bibinfo{year}{2012}):
  \emph{\bibinfo{title}{Asynchronous cellular automata and dynamical
  properties}}.
\newblock {\sl \bibinfo{journal}{Natural Computing}}
  \bibinfo{volume}{11}(\bibinfo{number}{2}), pp. \bibinfo{pages}{269--276},
  \doi{10.1007/s11047-012-9308-y}.

\bibitemdeclare{article}{Sorkin}
\bibitem{Sorkin}
\bibinfo{author}{Rafael \surnamestart Sorkin\surnameend}
  (\bibinfo{year}{1975}): \emph{\bibinfo{title}{Time-evolution problem in Regge
  calculus}}.
\newblock {\sl \bibinfo{journal}{Phys. Rev. D}} \bibinfo{volume}{12}, pp.
  \bibinfo{pages}{385--396}, \doi{10.1103/PhysRevD.12.385}.

\end{thebibliography}
\end{document}